\documentclass[a4paper,UKenglish,cleveref, autoref, thm-restate]{lipics-v2019}

\usepackage{amsmath,amssymb,amsthm,mathtools,comment}
\usepackage{enumerate}
\usepackage{float}
\usepackage{graphicx}
\usepackage{complexity}
\graphicspath{ {figures/} }
\usepackage{tikz}
\usetikzlibrary{automata,positioning, decorations.pathreplacing, arrows, decorations.markings}
\theoremstyle{definition}

\usepackage{microtype}

\title{Efficient Isolation of Perfect Matching in $O(\log n)$ Genus Bipartite Graphs} 

\titlerunning{Efficient Isolation of Perfect Matching in $O(\log n)$ Genus Bipartite Graphs} 

\author{Chetan Gupta}{Indian Institute of Technology, Roorkee, India \url{} }{chetan.gupta@cs.iitr.ac.in}{}{Ministry of Electronics and IT, India, VVS PhD program}

\author{Vimal Raj Sharma}{Indian Institute of Technology, Jodhpur, India 
\url{} }{vimalraj@iitj.ac.in}{}{Ministry of Electronics and IT, India, VVS PhD program}

\author{Raghunath Tewari}{Indian Institute of Technology, Kanpur, India
\url{} }{rtewari@cse.iitk.ac.in}{}{Young Faculty Research Fellowship, Ministry of Electronics and IT, India and INSPIRE Fellowship, Department of Science and Technology, India}

\authorrunning{C. Gupta, V. R. Sharma and R. Tewari} 

\Copyright{Chetan Gupta, Vimal Raj Sharma and Raghunath Tewari} 
\nolinenumbers

\ccsdesc[100]{Theory of computation~Complexity classes}
\ccsdesc[100]{Theory of computation~Computational complexity and cryptography}

\keywords{Logspace computation, High genus, Matching isolation} 
\hideLIPIcs
\category{} 

\relatedversion{} 

\supplement{}

\DeclareMathOperator{\sign}{sign}
\DeclareMathOperator{\Area}{Area}

\newclass{\ReachUL}{ReachUL}
\newclass{\coUL}{coUL}
\newcommand{\wpl}{w_{\textrm{pl}}}
\newcommand{\wfks}{w_{\textrm{fks}}}

\newcommand{\wcomb}{w_{\textrm{comb}}}

\newcommand{\wside}{w_{\textrm{side}}}
\newcommand{\tail}{\textrm{tail}}
\newcommand{\head}{\textrm{head}}
\newcommand{\wund}{w^{\textrm{{\tiny{und}}}}}
\newcommand{\wcomund}{w_{\textrm{comb}}^{\textrm{\tiny{und}}}}
\newcommand{\gpl}{G_{\text{planar}}}
\newcommand{\EC}{E_\mathcal{C}}

\begin{document}

\maketitle

\begin{abstract}
We show that given an embedding of an $O(\log n)$ genus bipartite graph, one can construct an edge weight function in logarithmic space, with respect to which the minimum weight perfect matching in the graph is unique, if one exists. 

As a consequence, we obtain that deciding whether such a graph has a perfect matching or not is in $\SPL$. In 1999, Reinhardt, Allender and Zhou proved that if one can construct a polynomially bounded weight function for a graph in logspace such that it isolates a minimum weight perfect matching in the graph, then the perfect matching problem can be solved in $\SPL$. In this paper, we give a deterministic logspace construction of such a weight function for $O(\log n)$ genus bipartite graphs.
\end{abstract}

\newpage

\section{Introduction}
Given a graph $G(V,E)$, a \textit{perfect matching} is defined as a set of disjoint edges which covers all the vertices in the graph. The perfect matching problem asks whether a graph has a perfect matching or not. The first polynomial time sequential algorithm to solve this problem was given by Edmonds \cite{Edmonds65}. Since then, there has been a lot of effort to solve this problem efficiently in a parallel computation model. $\NC$ is a class of problem that can be solved efficiently in parallel computation model. Lov{\'{a}}sz gave the first randomized $\NC$ algorithm to solve the perfect matching problem \cite{L79}. However, the question whether the problem can be solved in  $\NC$ or not is still open. 

Mulmuley et al. made significant progress in answering this question and gave the famous \textit{isolating lemma} \cite{MVV87}. 

\begin{lemma}(Isolating Lemma \cite{MVV87})
For a set $S = \{x_1,x_2, \ldots x_n\}$, let $F$ be a family of subsets of $S$. If the elements in the set $S$ are assigned integer weights chosen uniformly and independently from the set $\{1,2,\ldots 2n\}$ then with probability greater than half there is a unique minimum weight set in $F$.
\end{lemma}

Mulmuley et al. used this lemma to get a randomized $\NC$ algorithm for finding a perfect matching in graphs. They also showed that if one can construct an isolating weight function in $\NC$ (derandomizing the isolating lemma), then a perfect matching can be found in $\NC$. $\SPL$ is a class of problems reducible to computing determinant with the promise that the determinant is either 0 or 1. Allender et al. proved that if one can construct a perfect matching isolating weight function in logspace then the perfect matching problem can be solved in $\SPL$, which is a subset of $\NC^2$ \cite{ARZ99}. In a recent result, a quasi-polynomial ($O(\log^2n)$-bit) size isolating weight function was constructed for bipartite graphs which implies that the perfect matching problem can be solved in quasi-$\NC$ \cite{FGT16}. This result was subsequently extended to general graphs as well \cite{ST17}. However, constructing polynomially bounded isolating weight functions for general graphs has been elusive so far. Constructing isolating weight functions also has ramification in the \textit{directed graph reachability} problem. A logspace construction of a polynomially bounded path isolating weight function will imply that reachability problem in directed graphs can be solved in $\UL$, which will solve the $\NL$ vs. $\UL$ question, which has been open for a very long time\cite{RA00}. Also, a logspace construction of a polynomially bounded perfect matching isolating weight function even for bipartite graphs will prove that $ \NL \subseteq \SPL$ \cite{CSV84}.

Although constructing polynomially bounded isolating weight function seems to be hard for general graphs, such weight functions have been constructed for various subclasses of graphs such as planar graphs \cite{TV12}, bounded genus graphs \cite{DKTV11}, $K_{3,3}$ and $K_5$-free graphs \cite{AGGT16}, graph with small number of matchings \cite{GK87, AHT07} and graph with small number of nice cycles \cite{Hoang10}. The weight function constructed in \cite{DKTV11} is a $O(g \cdot \log n)$-bit weight function for $g$-genus graphs. Thus their result does not yield a polynomial size weight function for the graphs of genus more than constant. The question whether one can construct a polynomially bounded isolating weight function efficiently for graphs of genus beyond constant or not has been open since then. In this work, we settle this question by constructing a $O(g + \log n)$ bit isolating weight function for $g$-genus graphs. Thus our result gives a polynomial size isolating weight function for $O(\log n)$ genus bipartite graphs.


For a class of bipartite graphs, one way to obtain an isolating weight function is to construct a \textit{skew-symmetric} weight function for the same class of directed graphs such that every cycle in the graph gets a nonzero weight. This is the common technique in most of the above mentioned results. Having a skew-symmetric weight function such that it gives nonzero weights to every cycle in the graph, is sufficient for both path and matching isolation but is not necessary. Also, a weight function which isolates a path in the graph may not isolate a matching and vice-versa. That is why the weight functions constructed in \cite{KT16}, \cite{MP17} and \cite{GST19} are path isolating but do not isolate perfect matching. In this result, we construct a weight function which isolates a perfect matching in $g$-genus graphs even though it does not give nonzero weight to every cycle in the graphs.

\subsection{Our Result} 

In this paper, we extend the above line of work and prove the following theorem.
\begin{theorem}
\label{thm:main}
Given an undirected $O(\log n)$ genus bipartite graph along with its polygonal schema, the problem of deciding whether the graph has a perfect matching or not is in \SPL.
\end{theorem}

Given a $g$-genus bipartite graph $G$ we construct $O(g+\log n)$-bit weight functions $w_1,w_2, \ldots w_k$, where $k = O(n^c+ 2^{g})$, such that there exists a unique minimum weight perfect matching in the $G$ with respect to some $w_i$, if $G$ has a perfect matching. To achieve this, we first construct a directed graph $\vec{G}$ which is same as $G$, but its edges are assigned direction as follows. Let $L$ and $R$ be the two sets of the bipartition  of $G$. We assign a direction to all the edges in $\vec{G}$ from $L$ to $R$. Then we divide the perfect matchings of $\vec{G}$ into different classes according to their \textit{signatures}. Signature of a matching represents the parity of number of its edges crossing each pair of sides of the \textit{polygonal schema} of $\vec{G}$ (defined in Section \ref{sec:prelims}). Matchings in one class are said to be topologically equivalent to each other in a sense. Polygonal schema of a $g$-genus graph consists of $2g$ pairs of sides. Therefore there are $2^{2g}$ many classes. We construct our isolating weight function in two steps. In the first step, we construct a weight function which is a linear combination of the weight function constructed in \cite{TV12} and another weight function defined later in this paper and show that there is at most one minimum weight perfect matching in each class with respect to this weight function. In the second step, we use the hashing scheme of Fredman, Koml\'{o}s and Szemer\'{e}di \cite{FKS84} to get $k$ many weight functions $w_1,w_2,\ldots ,w_k$ such that for some $i\leq k$, $w_i$ isolates a minimum weight perfect matching in $\vec{G}$. A matching in $\vec{G}$ corresponds to a unique matching ${G}$ and vice-versa. Therefore we get a unique minimum weight perfect matching in $G$ with respect to $w_i$.

For $g= O(\log n)$ we get $k= O(n^{c'})$, for some constant $c'>0$. That means we get polynomially many weight functions such that there is at most one minimum weight perfect matching in the graph with respect to at least one of the weight function. Then we apply the result of \cite{ARZ99} to get an $\SPL$ algorithm for perfect matching problem in $O(\log n)$ genus bipartite graphs.

\textbf{Comparison with the path isolating weight function for $O (\log n)$ genus graphs \cite{GST19}:} Note that the weight function constructed in \cite{GST19} is also a linear combination of two weight functions, one of which gives nonzero weights  to all surface separating cycles in the graph. Therefore, when we divide the paths between a pair of vertices into classes and take any two minimum weight non-intersecting paths with respect to this weight function from the same class, we know that the cycle formed by reversing one of the paths is surface separating. Since every surface separating cycle has nonzero weight, and the weight function is skew-symmetric, this implies that these paths can not be of equal weights. Which means there is at most one minimum weight path in each class with respect to that weight function. Similarly, we handle the case when the paths are intersecting. However, that same weight function does not work here in matching isolation. Here also we first divide the matchings into classes according to their signatures. Now if we consider two minimum weight perfect matchings within a class, all the cycles formed by taking their disjoint union can be surface non-separating. Since the weight of a surface non-separating cycle can be zero with respect to that weight function, this does not give any contradiction to the fact that there can be two minimum weight perfect matchings within a class. In this paper, we overcome this hurdle by constructing a new weight function which isolates a matching within a class. Then we isolate a matching across the classes by the technique mentioned above. 

\subsection{Organization of the Paper}
The rest of the paper is organized as follows. In Section \ref{sec:prelims}, we define the necessary notations and a suitable representation of high genus graphs which we use in this paper. In Section \ref{sec:wtfunc}, we define the first part of our weight function, which is a linear combination of two weight functions defined in that section. In Section \ref{sec:isolation}, we prove that the number of minimum weight perfect matchings with respect to this weight function is very small. Then we use the hashing scheme of \cite{FKS84} to obtain our final weight function, which isolates a minimum weight perfect matching in the graph.

\section{Preliminaries and Notations}
\label{sec:prelims}
A $g$-genus surface is a sphere with $g$-many handles on it. A $g$-genus graph is a graph which can be embedded on a $g$-genus surface without intersecting its edges. A $g$-genus surface can be represented by a polygon called \textit{polygonal schema}(see Figure \ref{fig:schema}). The polygonal schema of a $g$-genus surface has $4g$-sides $T_1,T_2,T_1',T_2', \ldots , T_{2g-1}',T_{2g}'$ identified in pairs. The sides $T_i$ and $T_i'$ form a pair together in the sense that an edge going into $T_i$ will come out of $T_i'$ and vice versa. An embedding of a graph $G$ on a $g$-genus surface can be represented by an embedding of $G$ inside this polygon. In such an embedding an edge $\{u,v\}$ of a graph $G$ is said to cross a side $S$ of the polygonal schema, if $u$ or $v$ is incident on the side $S$ (for example in Figure \ref{fig:schema}, the edge $\{a,c\}$ is crossing the sides $T_1$ and $T_1'$). We assume that we are given the combinatorial embedding of the graph $G$ inside this polygon together with the ordered set of edges crossing each side of the polygon. We also assume that no vertex of $G$ lies on the sides of the polygonal schema. Such an embedding is called the \textit{polygonal schema of the graph $G$}.

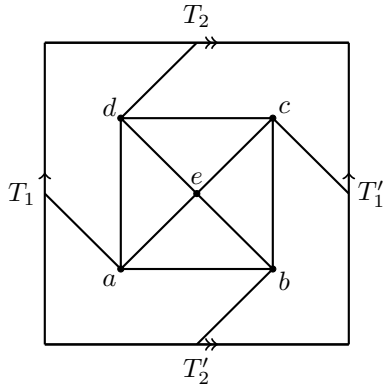
\begin{figure}
\begin{center}
\begin{tikzpicture}[]
\usetikzlibrary{decorations.markings}
\begin{scope}[very thick,decoration={
    markings,
    mark=at position 0.57 with {\arrow{>}}}
    ]
\draw[postaction={decorate}][draw, thick] (0,0) -- (4,0);
\draw[postaction={decorate}][draw, thick] (4,0) -- (4,4);
\draw[postaction={decorate}][draw, thick] (0,4) -- (4,4);
\draw[postaction={decorate}][draw, thick] (0,0) -- (0,4);
\end{scope}

\begin{scope}[very thick,decoration={
    markings,
    mark=at position 0.55 with {\arrow{>}}}
    ]
\draw[postaction={decorate}][draw, thick] (0,0) -- (4,0);
\draw[postaction={decorate}][draw, thick] (0,4) -- (4,4);
\end{scope}

\draw[][draw, thick] (1,1) -- (3,1);
\draw[][draw, thick] (3,1) -- (3,3);
\draw[][draw, thick] (3,3) -- (1,3);
\draw[][draw, thick] (1,3) -- (1,1);
\draw[][draw, thick] (1,1) -- (2,2);
\draw[][draw, thick] (3,1) -- (2,2);
\draw[][draw, thick] (3,3) -- (2,2);
\draw[][draw, thick] (1,3) -- (2,2);
\draw[][draw, thick] (1,1) -- (0,2);
\draw[][draw, thick] (3,3) -- (4,2);
\draw[][draw, thick] (3,1) -- (2,0);
\draw[][draw, thick] (1,3) -- (2,4);

\draw[fill] (2,2) circle [radius=0.04];
\draw[fill] (1,1) circle [radius=0.04];
\draw[fill] (1,3) circle [radius=0.04];
\draw[fill] (3,1) circle [radius=0.04];
\draw[fill] (3,3) circle [radius=0.04];

\node at (.85,.85) {$a$};
\node at (3.15,.85) {$b$};
\node at (3.15,3.15) {$c$};
\node at (.85,3.15) {$d$};
\node at (2,2.2) {$e$};

\node at (-.3,2) {$T_1$};
\node at (4.3,2) {$T_1'$};
\node at (2,4.35) {$T_2$};
\node at (2,-.35) {$T_2'$};
\end{tikzpicture} 
\caption{Polygonal schema of $K_5$, embedded on a surface of genus 1. Edges $\{a,c\}$ and $\{b,d\}$ are crossing the sides $T_1$ and $T_2$ respectively. Vertices $a$ and $c$ are said to be incident on the sides $T_1$ and $T_1'$ respectively.}
\label{fig:schema}
\end{center}
\end{figure}

In the polygonal schema of a graph $G$, the edges which do not cross any side of the polygonal schema, we call them \textit{planar edges}. Note that in the polygonal schema of a graph $G$, the subgraph induced by the planar edges of $G$, is a planar graph and we call this subgraph $\gpl$.

A \textit{piecewise straight-line embedding} of a planar graph is an embedding where all the vertices of the graph have integral coordinates and the edges are piecewise straight line segment connecting their two end points. Given a combinatorial embedding of a planar graph, a piecewise straight-line embedding of it can be constructed in logspace \cite{TV12}. Thus given a polygonal schema of a $g$-genus graph $G$, a piecewise straight-line embedding of $\gpl$ can be constructed in logspace. We will need such an embedding to construct our desirable weight function.

Given the polygonal schema of a $g$-genus graph $G$, we define the \textit{signature} of an edge $e$ in $G$, denoted as sign$(e)$, as a $2g$-bit binary string $b_1b_2 \ldots b_{2g}$, such that $b_i =1 $ if $e$ crosses $T_i$, otherwise $0$. Similarly, for any set of edges say $E = \{e_1,e_2, \ldots, e_k\}$, we define the signature of $E$ as, sign$(E)$  = sign$(e_1)$ $\oplus$ sign$(e_2) \oplus \ldots$ $\oplus$ sign$(e_k)$, where $\oplus$ represents the bitwise-XOR operator. Note that the $i$-th bit in the signature of a set $E$ represents the parity of the number of edges from that set, crossing the side $T_i$, i.e. if the number of edges in the set $E$, crossing the side $T_i$ are even then $i$-th bit in the sign$(E)$ will be $0$; otherwise it will be $1$.
Followings are the properties of signature that we will use later in this paper. For any set of edges $E_1,E_2,$ and $E_3$ of the graph $G$, we have
\begin{alphaenumerate}
\item \textit{commutativity}: sign$(E_1)$ $\oplus$ sign$(E_2)$ = sign$(E_2)$ $\oplus$ sign$(E_1)$,
\item \textit{associativity}: $\big($sign$(E_1)$ $\oplus$ sign$(E_2) \big)$ $\oplus$ sign$(E_3)$ = sign$(E_1)$ $\oplus$ $\big($sign$(E_2)$ $\oplus$ sign$(E_3)\big)$.
\end{alphaenumerate} 

 Without loss of generality assume that each edge crosses at most one pair of sides of the polygonal schema. If it crosses more than one pair of sides, we break it into multiple edges by inserting dummy vertices. To preserve matching, we always break an edge into an odd number of edges. Every term defined until now remains the same in case of directed graphs as well.

Since in this paper we work with both directed and undirected graphs, it is essential that we make a demarcation in the notation used for directed and undirected graphs. For a directed edge $\vec{e} = (u,v)$, the edge $e = \{u,v\}$ represents the underlying undirected edge and the edge $\vec{e}^{\,r}$ represents the directed edge $(v,u)$ that is the edge $\vec{e}$ with its direction reversed. Similarly, for any set of directed edges $\vec{E}$, set $E$ represents the set of underlying undirected edges of $\vec{E}$ and set $\vec{E}^{\,r}$ represents the set where each edge $\vec{e} \in \vec{E}$ is replaced with the edge $\vec{e}^{\,r}$.

In a directed graph $\vec{G}$, we call a set of edges $\vec{C}$, a \textit{directed cycle} if ($i$) edges of $C$ (underlying undirected edges of $\vec{C}$) form a simple cycle and, ($ii$) for every two adjacent edges of $\vec{C}$, tail of one edge is followed by the head of another edge. When we call $\vec{C}$ just a \textit{cycle} then $(ii)$ may not hold. Similarly, we can define a \textit{directed path} and \textit{path} in $\vec{G}$.

 $(0)^{k}$ represents the string $\overbrace{00\ldots0}^{k \text{-times}}$, where $k$ is an integer. For an integer $l>0$, $[l]$ denotes the set $\{1,2,\ldots, l\}$.

\section{Isolating Weight function}
\label{sec:wtfunc}

As discussed in the introduction, our main goal here is to construct a weight function for graphs efficiently. Let us first define the weight function formally. A weight function for a graph (directed or undirected) $G (V,{E})$ is a map $w: {E} \rightarrow Z$ which assigns an integer weight to every edge in the graph. For any set of edges $E'$ in the graph, the weight of the set $E'$ is defined as $w(E') = \sum_{e \in E'}w(e)$. A weight function $w$ for a graph $G$ is called \textit{min-isolating} if there exists at most one minimum weight perfect matching in $G$ with respect to the weight function $w$. 

In case of directed graphs, a weight function $w$ is called \textit{skew-symmetric} if $w(\vec{e}) = -w(\vec{e}^{\,r}),$ for all $\vec{e}\in \vec{E}$.

For a $g$-genus graph $\vec{G}$, we define a weight function $\wcomb$ which is a linear combination of the following two weight functions. 

\begin{itemize}
\item 
The first weight function we define is the same as the one defined  in \cite{TV12} for directed planar graphs. We call it $\wpl$. As we mentioned in Section \ref{sec:prelims}, we can construct a piecewise straight-line embedding of $\vec{G}_\text{planar}$ in logspace. In such an embedding an edge of the graph $\vec{G}_\text{planar}$ consists of constant many straight line segments. We assign weights to these line segments and the weight of an edge is defined as the sum of the weights of the line segments constituting that edge. Let $\vec{l}$ be a line segment such that $(x_1,y_1)$ and $(x_2,y_2)$ be the coordinates of its head and tail in such a piecewise straight-line embedding. Weight of $\vec{l}$ is defined as $\wpl(\vec{l}) = (y_2-y_1)(x_1+x_2)$
and weight of an edge $\vec{e}$ is defined as

\begin{eqnarray}
    \wpl (\vec{e})= 
\begin{cases}
    \sum_{\vec{l} \in \vec{e}} \wpl(\vec{l}),& \text{ if $\vec{e}$ is a planar edge,} \\
   0,& \text{ otherwise.}\\
\end{cases}
\end{eqnarray}

We state the following theorem regarding the weight function $\wpl$, which gives us a characterization of the weight of a directed cycle in a directed planar graph, established as Lemma $3$ in \cite{TV12}.\\
\begin{theorem} 
\label{thm:tv}
\cite{TV12}
Given a piecewise straight-line embedding of a planar graph $\vec{G}$, there exists a logspace computable weight function $\wpl$ such that for any directed cycle $\vec{C}$ in $\vec{G}$, we have $\wpl(\vec{C}) = 2 \cdot \Area (\vec{C})$ if $\vec{C}$ is a counter-clockwise cycle and $\wpl(\vec{C}) = - (2 \cdot \Area (\vec{C}))$ if $\vec{C}$ is a clockwise cycle, where $\Area(\vec{C})$ is the area of the region enclosed by $\vec{C}$.\\
\end{theorem}

\item We define another weight function $\wside$ as follows. Let $\sigma = (\vec{f_1}, \vec{f_2}, \ldots ,\vec{f_k})$ be the ordered set of edges crossing the sides of the polygonal schema $T_1$ to $T_{2g}$, ordered in a clockwise manner starting from the tail of $T_1$. 
 
\[
    \wside (\vec{f_i})= 
\begin{cases}
    i,& \text{if $\tail(\vec{f_i})$ is incident on some  side $T_j$ for $j \in [2g]$,} \\
   -i,& \text{if $\head(\vec{f_i})$ is incident on some  side $T_j$ for $j \in [2g]$.}
\end{cases}
\]
For all other edges $\vec{e}$, $\wside (\vec{e}) = 0$.


Our weight function $\wside$ is somewhat similar to the weight function defined in Theorem $8$ of \cite{DKTV11}. However, the main difference is that, in \cite{DKTV11}, they define $2g$ many weight functions (one for each pair of sides of the polygonal schema) similar to $\wside$ and their final weight function is a linear combination of those $2g$ weight functions, making it an $ O(g \cdot \log n)$-bit size weight function for $g$-genus graphs. Whereas in this paper $\wside$ is a single $O(\log n)$-bit weight function for a $g$-genus graph.

 Since each of these two weight functions are polynomially bounded and are computable in logspace, the overall computation remains in logspace as well. We combine these two weight functions into a single weight function and call it $\wcomb$, defined as follow:
\begin{eqnarray}
\wcomb = \wpl \cdot n^{10} +\wside.
\end{eqnarray}
\end{itemize}
Since for any two subsets of edges $\vec{E_1'}$ and $\vec{E_2'}$ of the graph, both weight functions $\wpl$ and $\wside$ are loosely bounded by $n^{10}$, hence $\wcomb(\vec{E_1'}) = \wcomb(\vec{E_2'})$ if and only if $\wpl(\vec{E_1'}) = \wpl(\vec{E_2'})$ and $\wside(\vec{E_1'}) = \wside(\vec{E_2'})$.\\

Note that in the perfect matching problem, we are given an undirected graph and asked to find if the graph has a perfect matching or not. However, we have defined the weight function $\wcomb$ for directed graphs. In order to give weights to an undirected bipartite graph $G$, we first obtain a directed graph $\vec{G}$ and construct a weight function for $\vec{G}$. Then we use that weight function to build a weight function for $G$.

Let $G$ be an undirected bipartite graph and $(L,R)$ be its bipartition. We construct a directed graph $\vec{G}$ as follows. For an edge $\{u,v\}$ in $G$ such that $u\in L$ and $v \in R$, we replace it with a directed edge $(u,v)$ in $\vec{G}$. We use Reingold's algorithm \cite{Reingold08} to find out whether a vertex  belongs to $L$ or $R$. Let $w$ be a weight function for $\vec{G}$. We define corresponding weight function $\wund$ for $G$ as follow. For an edge $\{u,v\}\in G$ such that $u \in L$ and $v \in R$,
\begin{equation}
  \wund \big( \{u,v\} \big)=w (u,v), \text{ where $(u,v)\in \vec{G}$}
\end{equation}
Note that if $\vec{M}$ is a matching of weight $t$ in $\vec{G}$ then $M$ will be a matching of weight $t$ in $G$. Thus, if $w$ is a min-isolating weight function for $\vec{G}$ then $\wund$ will be min-isolating for $G$. Also note that the construction of the graph $\vec{G}$ is the place where we use the bipartiteness of $G$ crucially. 
 
In the next section, we will construct a min-isolating weight function for directed $g$-genus bipartite graphs. Then ultimately we will use that weight function to obtain a min-isolating weight function for undirected $g$-genus bipartite graphs.

\section{Isolating a Minimum Weight Perfect Matching}
\label{sec:isolation}

Let $\vec{G}$ be a $g$-genus bipartite graph and $(L,R)$ be its bipartition. Let us assume that all the edges in $\vec{G}$ have direction from $L$ to $R$. We will prove that there are at most $2^{2g}$ minimum weight perfect matchings in $\vec{G}$ with respect to the weight function $\wcomb$, if $\vec{G}$ has a perfect matching. 

Let $\vec{M}$ be a perfect matching in $\vec{G}$. As we defined in Section \ref{sec:prelims}, the signature of $\vec{M}$ is,  
\begin{eqnarray*}
\text{sign$(\vec{M}) =$  sign$(\vec{e_1})$ $\oplus$ sign$(\vec{e_2})$ $\oplus \ldots \oplus $ sign$(\vec{e_j}),$ where $\vec{e_i} \in \vec{M}$ for all 	$i \in [j]$.} 
\end{eqnarray*}

Note that for a $g$-genus graph each matching has a $2g$-bit signature. Thus there are $2^{2g}$ many possible signatures. For each $0 \leq i \leq {2^{2g}-1}$, let bin($i$) represent the $2g$-bit binary number(with possible leading $0$'s) equivalent to an integer $i$. We define a class $A_i$ of perfect matchings in $\vec{G}$ with respect to the signature bin$(i)$ for all $0\leq i \leq {2^{2g}-1}$, as
\begin{eqnarray*}
A_i = \{ \vec{M} \mid \vec{M} \text{ is a perfect matching in $\vec{G}$ and sign$(\vec{M})=$ bin$(i)$}\}
\end{eqnarray*}

We will prove that there exists at most one minimum weight perfect matching in each class with respect to the weight function $\wcomb$.

\begin{lemma}
\label{lem:unimat}
For a $g$-genus bipartite graph $\vec{G}$, there exists at most one minimum weight perfect matching in the class $A_i$ with respect to the weight function $\wcomb$, for all $ 0 \leq i \leq 2^{2g}-1$.	
\end{lemma}

For two matchings $\vec{M_1}$ and $\vec{M_2}$ in $\vec{G}$, we define
$\vec{E}_{\vec{M_1} \Delta \vec{M_2}} = (\vec{M_1} \cup \vec{M_2}) \setminus (\vec{M_1} \cap \vec{M_2})$. Let us first prove the following lemma about the characterization of the edges in the set $\vec{E}_{\vec{M_1} \Delta \vec{M_2}}$, when $\vec{M_1}$ and $\vec{M_2}$ are two perfect matchings from the same class.

\begin{lemma}
\label{lem:twomatch}
If $\vec{M_1}$ and $\vec{M_2}$ are the two perfect
matchings in the class $A_i$ then $\sign(\vec{E}_{\vec{M_1} \Delta \vec{M_2}})=(0)^{2g}$ that is, the edges in the set $\vec{E}_{\vec{M_1} \Delta \vec{M_2}}$ collectively cross each side of the polygonal schema an even number of times.
\end{lemma}

\begin{proof}
Since $\vec{M_1}$ and $\vec{M_2}$ are the matchings from the same class, we have
\begin{eqnarray*}
\sign(\vec{M_1}) & =& \sign(\vec{M_2})\\
\sign(\vec{M_1}) \oplus  \sign(\vec{M_2}) &=& (0)^{2g} \\ 
\Big(\sign(\vec{M_1} \cap \vec{M_2}) \oplus \sign(\vec{E}_{\vec{M_1} \Delta \vec{M_2}} \setminus \vec{M_2}) \Big ) \oplus \\ \Big (\sign(\vec{M_1} \cap \vec{M_2}) \oplus \sign(\vec{E}_{\vec{M_1} \Delta \vec{M_2}} \setminus \vec{M_1}) \Big ) & =& (0)^{2g}
\end{eqnarray*}

 We know that for any set of edges $\vec{S}$, sign$(\vec{S})$ $\oplus$ sign$(\vec{S}) = (0)^{2g}$; and from the properties of signature mentioned in Section \ref{sec:prelims}, we have 
\begin{eqnarray*}
\big(\sign(\vec{E}_{\vec{M_1} \Delta \vec{M_2}} \setminus \vec{M_2}) \big ) \oplus  \big (\sign(\vec{E}_{\vec{M_1} \Delta \vec{M_2}} \setminus \vec{M_1}) \big )  &=& (0)^{2g} \\
\sign(\vec{E}_{\vec{M_1} \Delta \vec{M_2}})  &=& (0)^{2g}.
\end{eqnarray*}	
\end{proof}

We will now show that there is at most one minimum weight perfect matching in each class. Assume that $\vec{M_1}$ and $\vec{M_2}$ are the two minimum weight perfect matchings in the class $A_i$ with respect to the weight function $\wcomb$. We know that the edges in the set $\vec{E}_{\vec{M_1} \Delta \vec{M_2}}$ form vertex disjoint cycles. Let $\vec{C_1},\vec{C_2}, \ldots ,\vec{C_k}$ be those cycles. Notice that all the edges in the cycle $\vec{C}_i$ are directed from $L$ to $R$ therefore $\vec{C}_i$ is not a directed cycle, for any $i$. Also, note that each $\vec{C_i}$ consists of even number of edges and contain alternating edges from $\vec{M_1}$ and $\vec{M_2}$. Hence we can claim the following.

\begin{claim}\label{eqcycle}
Let $\vec{E}_{1i}$ and $\vec{E}_{2i}$ be the set of edges of $\vec{M_1}$ and $\vec{M_2}$ respectively in $\vec{C_i}$ then $\wcomb(\vec{E}_{1i}) = \wcomb (\vec{E}_{2i})$, for all $ i \in [k]$.
\end{claim} 

\begin{proof}
Let us assume that there exists some $j \in [k] $ such that $\wcomb(\vec{E}_{1j})  \neq \wcomb (\vec{E}_{2j})$. Without loss of generality assume that $\wcomb(\vec{E}_{1j})  > \wcomb (\vec{E}_{2j})$. Now consider a new perfect matching $\big ((\vec{M_1} \setminus \vec{E}_{1j}) \cup \vec{E}_{2j} \big )$. This matching has strictly lesser weight than $\vec{M_1}$, which is a contradiction because we have assumed that $\vec{M_1}$ is a minimum weight perfect matching.
\end{proof}

Now consider another graph $\vec{G'}$ which is same as $\vec{G}$ but direction of the edges belonging to $\vec{M_2}$ is reversed in $\vec{G'}$. Let $\vec{M_1'}$ and $\vec{M_2'}$ be the matchings in $\vec{G'}$ corresponding to the matchings $\vec{M_1}$ and $\vec{M_2}$ in $\vec{G}$, i.e. underlying undirected edges of matchings $\vec{M_1}$ and $\vec{M_2}$ are same as that of matchings $\vec{M_1'}$ and $\vec{M_2'}$ respectively. We know that the edges in the set $\vec{E}_{\vec{M_1'} \Delta \vec{M_2'}}$ will form vertex disjoint cycles. Let $\vec{C_1'},\vec{C_2'},\ldots ,\vec{C_k'}$ be those cycles and $\vec{E}_{1i}'$ and $\vec{E}_{2i}'$ be the edges of matching $\vec{M_1'}$ and $\vec{M_2'}$ respectively, in the cycle $\vec{C_i'}$. By claim \ref{eqcycle} we know that
\begin{eqnarray*}
\wcomb(\vec{E}_{1i}) &=& \wcomb (\vec{E}_{2i}), \text{ for all }i \in [k].
\end{eqnarray*}
Also $\vec{E}_{1i} =\vec{E}_{1i}'$ and ${\vec{E}_{2i}} = \vec{E}_{2i}'^{\,r}$, therefore
\begin{eqnarray*}
\wcomb(\vec{E}_{1i}') &=& \wcomb (\vec{E}_{2i}'^{\,r}) \text{, for all }i \in [k].
\end{eqnarray*}
Since $\wcomb$ is skew-symmetric, we have
\begin{eqnarray}
\wcomb(\vec{E}_{1i}') &=& -\wcomb (\vec{E}_{2i}'),\nonumber \\
\wcomb(\vec{E}_{1i}') &+& \wcomb(\vec{E}_{2i}') = 0, \nonumber \\
\wcomb(\vec{C_i'})&=&0, \text{ for all }i \in [k].
\label{eqn:cyzerowt}	
\end{eqnarray}

Note that the edges in the set $\vec{E}_{1i}'$ have direction from $L$ to $R$ and the edges in set $\vec{E}_{2i}'$ have direction from $R$ to $L$ therefore the cycles $\vec{C_1'},\vec{C_2'},\ldots ,\vec{C_k'}$ are the directed cycles in $\vec{G'}$. We will now prove that $\wcomb(\vec{C_i'}) \neq 0$ for some $i \in [k]$, which will be a contradiction with Equation \ref{eqn:cyzerowt}. 

Since changing the direction of an edge does not change its signature, by Lemma \ref{lem:twomatch} we know that $\sign(\vec{C_1'}) \oplus \sign(\vec{C_2'}) \oplus \ldots \oplus \sign(\vec{C_k'}) = (0)^{2g}$.

\begin{lemma}
\label{lem:nozerowtcy}
Let $\vec{G'}$ be a $g$-genus graph which contains directed cycles $\{\vec{C_1'},\vec{C_2'}, \ldots ,\vec{C_k'} \}$ such that  $\sign{(\vec{C_1'})} \oplus \sign{(\vec{C_2'})} \oplus \ldots \oplus \sign{(\vec{C_k'})} = (0)^{2g}$. Then there exists $i\in [k]$, such that $\wcomb(\vec{C_i'}) \neq 0$.
\end{lemma}

\begin{proof}
First consider the case, when no edge of the cycles $\{\vec{C_1'},\vec{C_2'}, \ldots ,\vec{C_k'} \}$  crosses any side of the polygonal schema. In that case each cycle $\vec{C_i'}$ is a planar cycle i.e. consists of only planar edges. By Theorem \ref{thm:tv} we know that $\wpl(\vec{C_i}) \neq 0$, which implies that $\wcomb(\vec{C_i}) \neq 0$ for all $i \in [k]$. Hence the lemma holds in this case. 

We will now prove the lemma for the case when some edges of the cycles $\{\vec{C_1'},\vec{C_2'}, \ldots ,\vec{C_k'} \}$ cross some sides of the polygonal schema. 


Let us consider a graph $G''$ such that edges of $G''$ are the underlying undirected edges of the cycles $(\vec{C_1'},\vec{C_2'}, \ldots , \vec{C_k'})$. Let $\mathcal{C} = ({C_1''},{C_2''}, \ldots ,{C_k''} )$ be the cycles in ${G''}$ corresponding to  cycles $(\vec{C_1'},\vec{C_2'}, \ldots , \vec{C_k'})$. We will construct another directed graph $\vec{G''}$ from $G''$(by assigning direction to the edges of $G''$) such that either $\vec{C_i''} = \vec{C_i'}$ or  $\vec{C_i''} = \vec{C_i'^r}$, for all $i\in[k]$. Let $\EC$ be the set of edges of the cycles in $\mathcal{C}$. We assign direction to the edges of $\EC$ in two steps. In the first step, we assign direction to only those edges of $\EC$ which are crossing some side of the polygonal schema. In the second step, we assign direction to the planar edges of $\EC$, based on the direction of the edges which were assigned direction in the first step. 

We know that all the cycles in $\mathcal{C}$ collectively cross each side of the polygonal schema an even number of times. Let $E =(e_1,e_2, \ldots e_{2l})$ for some integer $l>0$, be the edges in the set $\EC$, which cross some of the sides of the polygonal schema, indexed in clockwise order from $T_1$ to $T_{2g}$, starting from the tail of $T_1$. Without loss of generality assume that no two edges in $E$ share a vertex because if they do, we insert a dummy vertex in one of the edges to replace its end point so that our assumption holds. We will need this assumption to simplify our analysis.

\begin{itemize}
\setlength{\itemsep}{.2cm}
\item \textit{Step 1:} In this step, we assign direction to the edges in the set $E$. Let $e_i = \{u,v\}$ be an edge in $E$ such that $u$ and $v$ are incident on sides $T_j$ and $T_j'$ respectively, of the polygonal schema. We assign direction to $e_i \in E$ as follows:
\begin{itemize}
\item Assign direction to $e_i$ from $u$ to $v$, if $i$ is odd, i.e. assign direction to $e_i$ in such a way that $u$ becomes the tail of $\vec{e_i}$ and $v$ becomes the head of $\vec{e_i}$ in $\vec{G''}$.

\item Similarly, assign direction to $e_i$ from $v$ to $u$, if $i$ is even.\\
\end{itemize}

Before going to \textit{Step $2$}, let us make the following observations. Let $\vec{E} = (\vec{e_1},\vec{e_2} \ldots \vec{e_{2l}})$ be the edges in $\vec{G''}$ corresponding to edges in $E$ after \textit{Step 1}. Let $X = \{v_1,v_2, \ldots v_{4l}\}$ be the vertices of the edges of $\vec{E}$ ordered in a clockwise manner, according to their incidence on the side of the polygonal schema, starting from the tail of $T_1$(see Figure \ref{fig:ord}). Note that,

\begin{figure}
\begin{center}
\begin{tikzpicture}[]
\usetikzlibrary{decorations.markings}
\begin{scope}[very thick,decoration={
    markings,
    mark=at position 0.57 with {\arrow{<}}}
    ]

\draw[postaction={decorate}][draw, thick] (4,4) -- (4,0);

\draw[postaction={decorate}][draw, thick] (0,4) -- (0,0);
\end{scope}

\begin{scope}[very thick,decoration={
    markings,
    mark=at position 0.55 with {\arrow{>>}}}
    ]
\draw[postaction={decorate}][draw, thick] (0,0) -- (4,0);
\draw[postaction={decorate}][draw, thick] (0,4) -- (4,4);
\end{scope}


 \draw[<-] (0,1.5) to [] (1,1.5);
 \draw[] (1,1.5) to [] (1,2.5);
 \draw[<-] (1,2.5) to [] (0,2.5);
 \draw[->] (1.5, 3.5) to [] (1.5,4);
 \draw[] (2,3) to [] (1.5, 3.5);
 \draw[] (2.5,3.5) to [] (2,3);
 \draw[->] (2.5,4) to [] (2.5,3.5);
 
 \draw[->] (1.5,0) to [] (1.5,1);
 \draw[<-] (2.5,0) to [] (2.5,1);
 
 \draw[->] (3,2.5) to [] (4,2.5);
 \draw[<-] (3,1.5) to [] (4,1.5);

 \draw[] (2.5,1) to [] (3,1.5);
 \draw[] (1.5,1) to [] (2,2);
  \draw[] (2,2) to [] (3,2.5);
\draw[fill] (1,1.5) circle [radius=0.02];
\draw[fill] (1,2.5) circle [radius=0.02];
\draw[fill] (1.5,3.5) circle [radius=0.02];
\draw[fill] (2,3) circle [radius=0.02];
\draw[fill] (2.5,3.5) circle [radius=0.02];
\draw[fill] (1.5,1) circle [radius=0.02];
\draw[fill] (2.5,1) circle [radius=0.02];
\draw[fill] (3,1.5) circle [radius=0.02];
\draw[fill] (3,2.5) circle [radius=0.02];
\draw[fill] (2,2) circle [radius=0.02];
\node at (1,1.25) {$v_1$};
\node at (1.2,2.3) {$v_2$};
\node at (1.3,3.5) {$v_3$};
\node at (2.3,3.5) {$v_4$};
\node at (2.7,2.5) {$v_5$};
\node at (2.7,1.5) {$v_6$};
\node at (2.7,.8) {$v_7$};
\node at (1.7,.8) {$v_8$};

\node at (-.2,2) {$T_1$};
\node at (2,4.2) {$T_2$};
\node at (4.2,2) {$T_1'$};
\node at (2,-.22) {$T_2'$};
\node at (2,2.2) {$a$};

\end{tikzpicture} 
\caption{$(v_1,v_2,v_3,v_4,v_5,v_6,v_7,v_8)$ are the vertices of the edges which are crossing sides of the polygonal schema. Path $v_8av_5$ is a planar path. }
\label{fig:ord}
\end{center}
\end{figure}
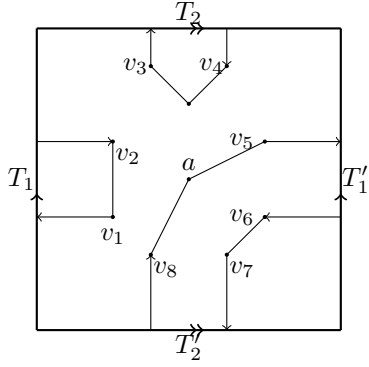

\begin{eqnarray}
\label{edge:prop}
\text{$\vec{e_i} = (v_{d_1},v_{d_2})$, where $d_1$ is odd and $d_2$ is even, for all $i \in [2l]$. }
\end{eqnarray}

We define a function $ \tau: X \rightarrow X $. $ \tau(v_i) = v_j$ if there is a simple path $P$ from $v_i$ to $v_j$ which consists of only planar edges of $\EC$, for $i,j \in[4l]$. We call such paths as \textit{planar paths} (see Figure \ref{fig:ord}). Since vertices in $X$ are the part of simple cycles, the function $\mathcal{\tau}$ is a bijective function.

\begin{lemma}
\label{lem:parity}
If $\tau(v_i) = v_j$, then $|i -j|$ is odd.
\end{lemma}
\begin{proof}
Assume that both $i$ and $j$ are odd. Without loss of generality assume that $j>i$. This implies that there are an odd number of vertices in the set $X$, between $v_i$ and $v_j$ namely, $X' = (v_{i+1}, v_{i+2}, \ldots v_{j-1})$. Note that vertices in $X'$ are part of non-intersecting simple cycles therefore they must be connected to each other through simple planar paths. Since $\tau$ is a bijective function we know that there is some vertex $v' \in X'$ such that $\tau(v') = v_t$ where $t \in [4l]$ and,  $t > j$ or $t<i$. This is not possible because it will imply that planar paths say from $v'$ to $v_t$ and from $v_i$ to $v_j$ say $P_1$ and $P_2$ respectively, must intersect each other. This is a contradiction since $P_1$ and $P_2$ are the parts of non-intersecting cycles. 
\end{proof}

\begin{lemma}
Let $P$ be a planar path between vertices $v_i$ and $v_j, i,j \in[4l]$. If $v_i$ is the head of some edge then $v_j$ will be the tail of some edge, in $\vec{E}$ and vice versa.
\end{lemma}
\begin{proof}
Let $v_i$ and $v_j$ both the vertices are the heads of the edges $e_{c_1}$ and $e_{c_2}$, where $c_1,c_2 \in  [k]$. We know that if $i$ is even then $j$ is odd and if $i$ is odd then $j$ is even. Without loss of generality assume that $i$ is even and $j$ is odd. However, from Equation \ref{edge:prop} we know that $j$ must be even. Hence we get a contradiction to Lemma \ref{lem:parity}. 

Similarly, we can handle the case when $v_i$ and $v_j$ are the tail of some edges.
\end{proof}

\item \textit{Step $2$:} Now we will assign the direction to the planar edges of $\vec{G''}$. This step is  straightforward. Take a planar path $P$ of $\vec{G''}$. Let $v'$ and $v''$ be its end vertices such that $v'$ is the head of an edge $\vec{e'}$ and $v''$ is the tail of some edge $\vec{e''}$, where $\vec{e'},\vec{e''}  \in \vec{E}$. Assign direction to all the edges in $P$ in such a way that the path $\vec{P'} = \vec{e'} \vec{P} \vec{e''} $ becomes a directed path in $\vec{G''}$.
\end{itemize}

Let $\vec{C_1''},\vec{C_2''}, \ldots , \vec{C_k''}$ be the cycles in $\vec{G''}$ after assigning direction to the underlying undirected cycles ${C_1''},{C_2''}, \ldots , {C_k''}$. After assigning direction using the above procedure, we can ensure that no two adjacent edges in the cycle $\vec{C_i''}$ for all $i\in [k]$ get opposite direction i.e. if $\vec{e}$ and $\vec{e'}$ are two adjacent edges in the cycle $\vec{C_i''}$ then the tail of $e$ will be followed by the head of $\vec{e'}$ or vice-versa (because of \textit{Step $2$}). This implies that $\vec{C_1''},\vec{C_2''}, \ldots , \vec{C_k''}$ are the directed cycles in $\vec{G''}$. Note that the way we have defined weight function $\wside$, we know that 
\begin{eqnarray*}
&& \wside(\vec{e}_i) < -\big(\wside(\vec{e}_{i+1})\big), \text{ for all odd $i \in [2l-1]$}\\
\Longrightarrow && \wside(\vec{e}_1)+\wside(\vec{e}_3)+ \ldots +\wside(\vec{e}_{2l-1}) < -\big(\wside(\vec{e}_2) +\wside(\vec{e}_4)+ \ldots +\wside(\vec{e}_{2l}) \big)\\
\Longrightarrow && \wside(\vec{e}_1)+\wside(\vec{e}_3)+ \ldots +\wside(\vec{e}_{2l-1})+ \wside(\vec{e}_2) +\wside(\vec{e}_4)+ \ldots +\wside(\vec{e}_{2l}) \neq 0.
\end{eqnarray*}
\\
Since for all planar edges $\vec{e}$, $\wside(\vec{e}) = 0$,
\begin{eqnarray*}
\sum_{i=1}^k \wside(\vec{C_i''}) \neq 0.
\end{eqnarray*}\\
Thus there exist some $i \in [k]$ such that
\begin{eqnarray}
\label{eqn:nonzero}
\wside(\vec{C_i''}) \neq 0  \Longrightarrow \wcomb(\vec{C_i''}) \neq 0,
\end{eqnarray}

Note that $\vec{C_i'}$ and $\vec{C_i''}$ for all $i \in [k]$, are the directed cycles such that their underlying undirected cycle is same. In a directed cycle there are only two directions possible. Therefore, we can say that
\begin{eqnarray}
\label{eqn:samecy}
&& \vec{C_i'} = \vec{C_i''} \text { or } \vec{C_i'} = \vec{C_i''^{\,r}},\nonumber\\
\Rightarrow &&\wcomb(\vec{C_i'}) = \wcomb(\vec{C_i''}) \text { or } \wcomb(\vec{C_i'}) = \wcomb(\vec{C_i''^{\,r}}), \nonumber \\
\Rightarrow && \wcomb(\vec{C_i'}) =   \wcomb(\vec{C_i''}) \text{ or } \wcomb(\vec{C_i'})= -\wcomb(\vec{C_i''}) \text{, for all $i \in [k],$ since} \\
&&\text{ $\wcomb$ is skew-symmetric.}\nonumber 
\end{eqnarray}
 
From Equation \ref{eqn:nonzero} and \ref{eqn:samecy} we can conclude that there exists some $ i \in [k]$ such that $\wcomb(\vec{C_i'}) \neq 0$. 
\end{proof}

This gives a contradiction with Equation \ref{eqn:cyzerowt}. Therefore we can conclude that there cannot exist two minimum weight perfect matchings in $\vec{G}$ within a class $A_i$ for all $ 0 \leq i \leq  2^{2g}-1$. This finishes the proof of Lemma \ref{lem:unimat}.

Note that we have proved that there is at most one minimum weight perfect matching in each class and there are total $2^{2g}$ many classes. Therefore, we can say that there are at most $2^{2g}$ minimum weight matchings in $\vec{G}$ with respect to the weight function $\wcomb$. As we mentioned in Section \ref{sec:wtfunc} that given a weight function $\wcomb$ for a directed bipartite graph $\vec{G}$ such that edges of $\vec{G}$ are directed from $L$ to $R$, we can get a weight function $\wcomund$ for underlying undirected graph $G$ such that if $\vec{M}$ is a matching of weight $t$ in $\vec{G}$ then $M$ will be a matching of weight $t$ in $G$.

\begin{lemma}
Given a $g$-genus graph $G$ along with its polygonal schema we can construct a weight function $\wcomund$ for $G$ in logspace such that there are at most $2^{2g}$ minimum weight perfect matchings in $G$ with respect to $\wcomund$.

\end{lemma}

Now that given an undirected graph $G$ we have obtained at most $2^{2g}$ many minimum weight perfect matchings in $G$, we will use the following hashing scheme by  Fredman, Koml\'{o}s and Szemer\'{e}di \cite{FKS84} to isolate a minimum weight perfect matching among them. Let us first state their result in a form suitable to our purpose. 

\begin{theorem} 
\label{thm:fks}
\cite{FKS84}
Let $S = \{x_1, x_2, \dots, x_k\}$ be a set of $n$-bit integers. Then there exists a ${O}(\log n + \log k)$-bit prime number $p$ so that for all $x_i \neq x_j \in S$, $x_i \bmod{p} \neq x_j \bmod{p}$.
\end{theorem}

Let $\mathcal{M}$ be the set of minimum weight perfect matchings in $G$ with respect to $\wcomund$. Assume edges of the graph $G$ are indexed as $e_1,e_2, \ldots, e_m$. Let $w_b$ be a weight function that assigns weight $2^i$ to the edge $e_i$.  This is an $m$-bit weight function, where $m \leq n^2$. All matchings in $G$ get different weight with respect to this weight function therefore, any two matchings $M_1,M_2 \in \mathcal{M}$, $w_b(M_1) \neq w_b(M_2)$. Also, note that $|\mathcal{M}| \leq 2^{2g}$, because each class has at most one minimum weight perfect matching. Thus by Theorem \ref{thm:fks} there exists an $O(\log n+g)$-bit prime $p$ such that with respect to weight function $\wfks \coloneqq$ $w_b$ mod $p$, every matching in $\mathcal{M}$ gets a different weight. Hence our final min-isolating weight function for $G$ will be,
\begin{eqnarray*}
w_p \coloneqq \wcomund \cdot n^{10} + \wfks,
\end{eqnarray*} 

Note that for every  ${O}(\log n + g)$-bit prime $p$ we get a corresponding weight function $w_p$ and by Theorem \ref{thm:fks} we know that there will be at least one ${O}(\log n + g)$-bit prime $p_1$ such that $w_{p_1}$ isolates a minimum weight perfect matching in $G$. Thus we can conclude the following.

\begin{theorem}
\label{thm:upm}
Given a $g$-genus graph along with its polygonal schema, we can construct weight functions $w_1, w_2, \ldots ,w_k$ in $O (\log n + g)$ space such that if graph has a perfect matching then for some $i \in [k]$ and, G has a unique perfect matching $M$ of weight $j$ with respect to weight function $w_i$, where $j,k = O(n^c + 2^{g})$ for some constant $c>0$. 
\end{theorem}
For a graph of genus $g=O(\log n)$ we get polynomially many weight functions $w_1,w_2, \ldots w_t$ where $t = O(n^c)$ for some constant c, such that each $w_i$ is polynomially bounded and there is a unique minimum weight perfect matching in graph with respect to at least one $w_i$ if $G$ has a perfect matching. Then we apply the algorithm given in \cite{ARZ99} to get an $\SPL$ algorithm for perfect matching in $O(\log n)$ genus bipartite graphs. This finishes the proof of Theorem \ref{thm:main}.



\bibliography{references}

\appendix

\end{document}